\documentclass[proceedings]{stacs}
\stacsheading{2010}{167-178}{Nancy, France}
 \firstpageno{167}
\usepackage{amsmath,amssymb,times}

\theoremstyle{definition}\newtheorem{claim}[thm]{Claim}

\begin{document}

\title{Two-phase algorithms for the parametric shortest path problem}

\author[lab1]{S. Chakraborty}{Sourav Chakraborty}
\address[lab1]{ CWI, Amsterdam, Netherlands}  
\email{sourav.chakraborty@cwi.nl}  
\thanks{Part of the research done when the first author was a postdoc in
Technion.}

\thanks{For the second author research was supported by an ERC-2007-StG grant number 202405 and
by an ISF grant 1011/06.}

\author[lab2]{E. Fischer}{Eldar Fischer}
\address[lab2]{Department of Computer Science, Technion, Haifa 32000, Israel}	
\email{eldar@cs.technion.ac.il}  

\author[lab3]{O. Lachish}{Oded Lachish}
\address[lab3]{Centre for Discrete Mathematics and its Applications, University of
Warwick , Coventry, UK}	
\email{oded@dcs.warwick.ac.uk}  
\thanks{Third author was supported in part by EPSRC award EP/G064679/1 and by the Centre for Discrete Mathematics and its Applications (DIMAP),
EPSRC award EP/D063191/1.}

\author[lab4]{R. Yuster}{Raphael Yuster}
\address[lab4]{Department of Mathematics, University of Haifa, Haifa 31905, Israel}	
\email{raphy@math.haifa.ac.il}  



\begin{abstract} A {\em parametric weighted graph} is a graph whose
edges are labeled with continuous real functions of a single common
variable.  For any instantiation of the variable, one obtains a
standard edge-weighted graph.  Parametric weighted graph problems are
generalizations of weighted graph problems, and arise in various
natural scenarios.  Parametric weighted graph algorithms consist of
two phases. A {\em preprocessing phase} whose input is a parametric
weighted graph, and whose output is a data structure, the advice, that
is later used by the {\em instantiation phase}, where a specific value
for the variable is given. The instantiation phase outputs the
solution to the (standard) weighted graph problem that arises from the
instantiation.  The goal is to have the running time of the
instantiation phase supersede the running time of any algorithm that
solves the weighted graph problem from scratch, by taking advantage of
the advice.

In this paper we construct several parametric algorithms for the
shortest path problem.  For the case of linear function weights we
present an algorithm for the single source shortest path problem. Its
preprocessing phase runs in $\tilde{O}(V^4)$ time, while its
instantiation phase runs in only $O(E+V \log V)$ time. The fastest
standard algorithm for single source shortest path runs in $O(VE)$
time. For the case of weight functions defined by degree $d$
polynomials, we present an algorithm with quasi-polynomial
preprocessing time $O(V^{(1 + \log f(d))\log V})$ and instantiation
time only $\tilde{O}(V)$.  In fact, for any pair of vertices $u,v$,
the instantiation phase computes the distance from $u$ to $v$ in only
$O(\log^2 V)$ time.  Finally, for linear function weights, we present
a randomized algorithm whose preprocessing time is
$\tilde{O}(V^{3.5})$ and so that for any pair of vertices $u,v$ and
any instantiation variable, the instantiation phase computes, in
$O(1)$ time, a length of a path from $u$ to $v$ that is at most
(additively) $\epsilon$ larger than the length of a shortest path.  In
particular, an all-pairs shortest path solution, up to an additive
constant error, can be computed in $O(V^2)$ time.

\end{abstract}

\maketitle

\section{Introduction}

In networking or telecommunications the search for the minimum-delay
path (that is the shortest path between two points) is always on.
The cost on  each edge, that is the time taken for a signal to
travel between two  adjacent nodes of the network, is often a
function of real time. Hence the shortest path between any two nodes
changes with time. Of course one can run a shortest path algorithm
every time a signal has to be sent, but usually some prior knowledge of
the network graph is given in advance, such as the structure of the
network graph and the cost functions on each edge (with time as a
variable).

How can one benefit from this extra information? One plausible way
is to preprocess the initial information and store the
preprocessed information. Every time the rest of the input is given,
using the preprocessed information, one can solve the optimization
problem  faster than solving the problem from scratch. Even if the
preprocessing step is expensive one would benefit by saving precious
time whenever the optimal solution has to be computed. Also, if the
same preprocessed information is used multiple times then the total
amount of resources used will be less in the long run.

Similar phenomena can be observed in various other combinatorial
optimization problems that arise in practice; that is, a part of the
input does not change with time and is known in advance. However, many
times it is hard to make use of this extra information.

In this paper we consider only those problems where the whole input is
a weighted graph. We assume that the graph structure and some
knowledge of how the weights on the edges are generated are known in
advance. We call this the {\it function-weighted graph} -- it is a graph
whose edges are labeled with continuous real functions. When all the
functions are univariate (and all have the same variable), the graph
is called a {\em parametric weighted graph}. In other words, the graph
is $G=(V,E,W)$ where $W:E \rightarrow \mathcal{F}$ and $\mathcal{F}$ is the
space of all real continuous functions with the variable $x$. If $G$
is a parametric weighted graph, and $r \in \mathbb{R}$ is any real
number, then $G(r)$ is the standard weighted graph where the weight of
an edge $e$ is defined to be $(W(e))(r)$.  We say that $G(r)$ is an
{\em instantiation} of $G$, since the variable $x$ in each function is
instantiated by the value $r$.  Parametric weighted graphs are
therefore, a generic instance of infinitely many instances of weighted
graphs.

The idea is to use the generic instance $G$ to precompute some general
generic information $I(G)$, such that for any given
instantiation $G(r)$, we will be able to
use the precomputed information $I(G)$ in order to speed up the time
to solve the given problem on $G(r)$, faster than just solving the
problem on $G(r)$ from scratch.  Let us make this notion more
precise.

A {\em parametric weighted graph algorithm} (or, for brevity, a {\em
parametric algorithm}) consists of two phases. A {\em preprocessing
phase} whose input is a parametric weighted graph $G$, and whose
output is a data structure (the advice) that is later used by the {\em
instantiation phase}, where a specific value $r$ for the variable is
given. The instantiation phase outputs the solution to the (standard)
weighted graph problem on the weighted graph $G(r)$. Naturally, the
goal is to have the running time of the instantiation phase
significantly smaller than the running time of any algorithm that
solves the weighted graph problem from scratch, by taking advantage of
the advice constructed in the preprocessing phase.  Parametric
algorithms are therefore evaluated by a pair of running times, the
{\em preprocessing time} and the {\em instantiation time}.

In this paper we show that parametric algorithms are beneficial for
one of the most natural combinatorial optimization problems: the {\em
shortest path} problem in directed graphs. Recall that given a
directed real-weighted graph $G$, and two vertices $u,v$ of $G$, the
distance from $u$ to $v$, denoted by $\delta(u,v)$, is the length of a
shortest path from $u$ to $v$. The {\em single pair} shortest path
problem seeks to compute $\delta(u,v)$ and construct a shortest path
from $u$ to $v$. Likewise, the {\em single source} shortest path
problem seeks to compute the  distances and shortest paths from a
given vertex to all other vertices, and the {\em all pairs} version
seeks to compute distances and shortest paths between all ordered
pairs of vertices.  In some of our algorithms we forgo the calculation
of the path itself to achieve a shorter instantiation time. In all
those cases the algorithms can be easily modified to also output a
shortest path, in which case their instantiation time is the sum of
the time it takes to calculate the distance and a time linear in the
size of the path to be output.

Our first algorithm is a parametric algorithm for single source
shortest path, in the case where the weights are {\em linear}
functions. That is, each edge $e$ is labeled with a function
$a_ex+b_e$ where $a_e$ and $b_e$ are reals. Such linear
parametrization has practical importance.  Indeed, in many problems
the cost of an edge is composed from some constant term plus a term
which is a factor of some commodity, whose cost varies (e.g. bank
commissions, taxi fares, vehicle maintenance costs, and so on).  Our
parametric algorithm has preprocessing time $\tilde{O}(n^4)$ and
instantiation time $O(m+n\log n)$ (throughout this paper $n$ and $m$
denote the number of vertices and edges of a graph, respectively). We
note that the fastest algorithm for the single source shortest path in
real weighted directed graphs requires $O(nm)$ time; the Bellman-Ford
algorithm \cite{Be-1958}.  The idea of our preprocessing stage is to
precompute some other linear functions, on the {\em vertices}, so that
for every instantiation $r$, one can quickly determine whether $G(r)$
has a negative cycle and otherwise use these functions to quickly
produce a reweighing of the graph so as to obtain only nonnegative
weights similar to the weights obtained by Johnson's algorithm
\cite{Jo-1977}. In other words, we {\em avoid} the need to run the
Bellman-Ford algorithm in the instantiation phase. The
$\tilde{O}(n^4)$ time in the preprocessing phase comes from the use of Megiddo's\cite{M-79} technique that we need in order to compute the linear vertex functions.

\begin{theorem}\label{thm:neg} There exists a parametric algorithm for
single source shortest path in graphs weighted by linear functions,
whose preprocessing time is $\tilde{O}(n^4)$ and whose instantiation
time is $O(m + n\log n)$.
\end{theorem}

Our next algorithm applies to a more general setting where the weights
are polynomials of degree at most $d$. Furthermore, in this case our
goal is to have the instantiation phase answering distance queries
between any two vertices in {\em sublinear} time. Notice first that if
we allow exponential preprocessing time, this goal can be easily
achieved. This is not hard to see since the overall
possible number of shortest paths (when $x$ varies over the reals) is
$O(n!)$, or from Fredman's decision tree for shortest paths whose
height is $O(n^{2.5})$ \cite{Fr-1976}. But can we settle for {\em
sub-exponential} preprocessing time and still be able to have
sublinear instantiation time? Our next result achieves this goal.

\begin{theorem}\label{thm:gen} There exists a parametric algorithm for
the single pair shortest path problem in graphs weighted by degree $d$
polynomials, whose preprocessing time is $O(n^{(O(1) + \log f(d))\log
n})$ and instantiation time $O(\log^2 n)$, where $f(d)$ is the time
required to compute the intersection points of two degree $d$
polynomials. The size of the advice that the preprocessing algorithm
produces is $O(n^{(O(1) + \log d)\log n})$.
\end{theorem}

The above result falls in the subject of sensitivity analysis where
one is interested in studying the effect on the optimal solution as the
value of the parameter changes. We give a linear-time (linear in the
output size) algorithm that computes the breaking points.

The practical and theoretical importance of shortest path problems
lead several researchers to consider fast algorithms that settle for
an approximate shortest path. For the general case (of real weighted
digraphs) most of the algorithms guarantee an {\em $\alpha$-stretch}
factor. Namely, they compute a path whose length is at most $\alpha
\delta(u,v)$. We mention here the $(1+\epsilon)$-stretch algorithm
of Zwick for the all-pairs shortest path problem, that runs in
$\tilde{O}(n^\omega)$ time when the weights are non-negative reals
\cite{Zw-2002}. Here $\omega < 2.376$ is the matrix multiplication
exponent \cite{CoWi-1990}.

Here we consider probabilistic additive-approximation algorithms, or
{\em surplus} algorithms, that work for linear weights which may have
positive and negative values (as long as there is no negative weight
cycle).  We say that a shortest path algorithm has an
$\epsilon$-surplus if it computes paths whose lengths are at most
$\delta(u,v)+\epsilon$.  We are unaware of any truly subcubic algorithm
that guarantees an $\epsilon$-surplus approximation, and which
outperforms the fastest general all-pairs shortest path algorithm
\cite{Ch-2007}.

In the linear-parametric setting, it is easy to obtain
$\epsilon$-surplus parametric algorithms whose preprocessing time is
$O(n^4)$ time, and whose instantiation time, for any ordered pair of
queried vertices $u,v$ is constant. It is assumed instantiations are
taken from some interval $I$ whose length is independent of $n$.
Indeed, we can partition $I$ into $O(n)$ subintervals $I_1,I_2,\ldots$
of size $O(1/n)$ each, and solve, in cubic time (say, using
\cite{Fl-1962}), the exact all-pairs solution for any instantiation
$r$ that is an endpoint of two consecutive intervals.  Then, given any
$r \in I_j=(a_j,b_j)$, we simply look at the solution for $b_j$ and
notice that we are (additively) off from the right answer only by
$O(1)$. Standard scaling arguments can make the surplus smaller than
$\epsilon$.  But do we really need to spend $O(n^4)$ time for
preprocessing? In other words, can we invest (significantly) less than
$O(n^4)$ time and still be able to answer instantiated distance
queries in $O(1)$ time? The following result gives a positive answer
to this question.

\begin{theorem}\label{thm:approx} Let $\epsilon > 0$, let
$[\alpha,\beta]$ be any fixed interval and let $\gamma$ be a fixed
constant. Suppose $G$ is a linear-parametric graph that has no
negative weight cycles in the interval $[\alpha,\beta]$, and for which
every edge weight $a_e+xb_e$ satisfies $|a_e|\leq\gamma$. There is a
parametric randomized algorithm for the $\epsilon$-surplus shortest
path problem, whose preprocessing time is $\tilde{O}(n^{3.5})$ and
whose instantiation time is $O(1)$ for a single pair, and hence
$O(n^2)$ for all pairs.
\end{theorem} We note that this algorithm works in the restricted
addition-comparison model.  We also note that given an ordered pair
$u,v$ and $r \in [\alpha,\beta]$, the algorithm outputs, in $O(1)$
time, a weight of an actual path from $u$ to $v$ in $G(r)$, and points
to a linked list representing that path.  Naturally, if one wants to
output the vertices of this path then the time for this is linear in
the length of the path.

The rest of this paper is organized as follows. The next subsection
shortly surveys related research on parametric shortest path
problems. In the three sections following it we prove Theorems
\ref{thm:neg}, \ref{thm:gen} and \ref{thm:approx}.
Section~\ref{sec:conclusion} contains some concluding remarks and open
problems.

\subsection{Related research} Several researchers have considered parametric
versions of combinatorial optimization problems. In particular
function-weighted graphs (under different names) have been
extensively studied in the subject of sensitivity analysis (see
\cite{vHKRW89}) where they study the effect on the optimal solution
as the parameter value changes.

Murty~\cite{Mur80} showed that for parametric linear programming
problems the optimal solution can change exponentially many times
(exponential in the number of variables). Subsequently, Carstensen
\cite{Ca-1983} has shown that there are constructions  for which
the number of shortest path changes while $x$ varies over the reals
is $n^{\Omega(\log n)}$. In fact, in her example each linear
function is of the form $a_e+xb_e$ and both $a_e$ and $b_e$ are
positive, and $x$ varies in $[0,\infty]$. Carstensen also proved
that this is tight. In other words, for any linear-parametric graph
the number of changes in the shortest paths is $n^{O(\log n)}$. A
simpler proof was obtained by Nikolova et al. \cite{Ni-Mi-2006},
that also supply an $n^{O(\log n)}$ time algorithm to compute the
path breakpoints. Their method, however, does not apply to the case
where the functions are not linear, such as in the case of degree
$d$ polynomials. Gusfield~\cite{gus} also gave a proof for the upper bound
of the number of breakpoints in the linear function version of the parametric shortest
path problem, in addition to studying a number of other parametric
problems.

Karp and Orlin \cite{KaOr-1981}, and, later, Young, Tarjan, and
Orlin~\cite{Yo-Or-1991} considered a special case of the
linear-parametric shortest path problem.  In their case, each edge
weight $e$ is either some fixed constant $b_e$ or is of the form
$b_e-x$. It is not too difficult to prove that for any given vertex
$v$, when $x$ varies from $-\infty$ to the largest $x_0$ for which
$G(x_0)$ has no negative weight cycle (possibly $x_0=\infty$), then
there are at most $O(n^2)$ distinct shortest path trees from $v$ to
all other vertices. Namely, for each $r \in [-\infty,x_0]$ one of
the trees in this family is a solution for single-source shortest
path in $G(r)$. The results in \cite{KaOr-1981,Yo-Or-1991} cleverly
and compactly compute all these trees, and the latter does it in
$O(nm+n^2\log n)$ time.

\section{Proof of Theorem~\ref{thm:neg}}

The proof of Theorem~\ref{thm:neg} follows from the following two
lemmas.

\begin{lemma}\label{lem:negcyl} Given a linear-weighted graph
$G=(V,E,W)$, there exist $\alpha, \beta \in \mathbb{R}\cup \{-\infty\}
\cup \{+\infty\}$ such that $G(r)$ has no negative cycles if and only
if $\alpha \leq r \leq \beta$.  Moreover $\alpha$ and $\beta$ can be
found in $\tilde{O}(n^4)$ time.
\end{lemma}

\begin{lemma}\label{lem:reweight} Let $G=(V,E,W)$ be a linear-weighted
graph. Also let $\alpha, \beta \in \mathbb{R}\cup \{-\infty\} \cup
\{+\infty\}$ be such that at least one of them is finite and for all
$\alpha \geq r \geq \beta$ the graph $G(r)$ has no negative cycle.
Then for every vertex $v \in V$ there exists a linear function
$g^{[\alpha,\beta]}_v$ such that if the new weight function $W'$ is
given by
$$
W'\left((u,v)\right) = W\left((u,v)\right) + g^{[\alpha, \beta]}_u -
g^{[\alpha, \beta]}_v
$$
then the new linear-weighted graph $G'=(V,E,W')$ has the property that
for any real $\alpha \leq r \leq \beta$ all the edges in $G'(r)$ are
non-negative.  Moreover the functions $g^{[\alpha, \beta]}_v$ for all
$v \in V$ can be found in $O(mn)$ time.
\end{lemma}

So given a linear-weighted graph $G$, we first use Lemma
\ref{lem:negcyl} to compute $\alpha$ and $\beta$.  If at least one of
$\alpha$ and $\beta$ is finite then using Lemma~\ref{lem:reweight} we
compute the $n$ linear functions $g^{[\alpha,\beta]}_v$, one for each
$v \in V$.  If $\alpha = -\infty$ and $\beta = +\infty$, then using
Lemma~\ref{lem:reweight} we compute the $2n$ linear functions
$g^{[\alpha, 0]}_v$ and $g^{[0,\beta]}_v$. These linear functions will
be the advice that the preprocessing algorithm produces.  The above
lemmas guarantee us that the advice can be computed in time
$\tilde{O}(n^4)$, that is the preprocessing time is $\tilde{O}(n^4)$.

Now when computing the single source shortest path problem from vertex
$v$ for the graph $G(r)$ our algorithm proceeds as follows:
\begin{enumerate}
\item If $r<\alpha$ or $r>\beta$ output ``$-\infty$'' as there exists
a negative cycle (such instances are considered invalid).
\item If $\alpha\leq r\leq \beta$ and at least one of $\alpha$ or
$\beta$ is finite then compute $g_u(r)$ for all $u \in V$. Use these
to re-weight the edges in the graph as in Johnson's algorithm
\cite{Jo-1977}.  If $\alpha = -\infty$ and $\beta = +\infty$ then if
$r\leq 0$ compute $g^{[\alpha,0]}_u(r)$ for all $u \in V$ and if
$r\geq 0$ compute $g^{[0,\beta]}_u(r)$ for all $u \in V$.  Notice that
after the reweighing we have an instance of $G'(r)$.
\item Use Dijkstra's algorithm \cite{Di-1959} to solve the single
source shortest path problem in $G'(r)$. Dijkstra's algorithm applies
since $G'(r)$ has no negative weight edges.  The shortest paths tree
returned by Dijkstra's algorithms applied to $G'(r)$ is also the
shortest paths tree in $G(r)$. As in Johnson's algorithm, we use the
results $d'(v,u)$ of $G'(r)$ to deduce $d(v,u)$ in $G(r)$ since, by
Lemma \ref{lem:reweight} $d(v,u)=d'(v,u)-g_v(r)+g_u(r)$.
\end{enumerate} The running time of the instantiation phase is
dominated by the running time of Dijkstra's algorithm which is $O(m +
n\log n)$ \cite{FrTa-1987}.

\subsection{Proof of Lemma~\ref{lem:negcyl}}

Since the weight on the edges of the graph $G$ are linear functions,
we have that the weight of any directed cycle in the graph is also a
linear function.  Let $C_1, C_2, \ldots, C_T$ be the set of all
directed cycles in the graph.  The linear weight function of a cycle
$C_i$ will be denoted by $\mbox{wt}(C_i)$.  If $\mbox{wt}(C_i)$ is not the
constant function, then let $\gamma_i$ be the real number for which
the linear equation $\mbox{wt}(C_i)$ evaluates to $0$.

\noindent Let $\alpha$ and $\beta$ be defined as follows:
$$
\alpha = \max_i \left\{ \gamma_i \mid \mbox{ $\mbox{wt}(C_i)$ has a positive
slope}\right\}.
$$
$$
\beta = \min_i \left\{ \gamma_i \mid \mbox{ $\mbox{wt}(C_i)$ has a negative
slope}\right\}.
$$
Note that if $\mbox{wt}(C_i)$ has a positive slope then
$\gamma_i = \min_x \left\{ \mbox{wt}(C_i)(x) \geq 0\right\}.$
Thus for all $x \geq \gamma_i$ the value of $\mbox{wt}(C_i)$ evaluated at
$x$ is non-negative. So by definition for all $x \geq \alpha$ the
value of the $\mbox{wt}(C_i)$ is non-negative if the slope of $\mbox{wt}(C_i)$ is
positive, and for any $x<\alpha$ there exists a cycle $C_i$ such that
$\mbox{wt}(C_i)$ has positive slope and $\mbox{wt}(C_i)(x)$ is negative.
Similarly, for all $x \leq \beta$ the value of the $\mbox{wt}(C_i)$ is
non-negative if the slope of $\mbox{wt}(C_i)$ is negative and for any
$x>\beta$ there exists a cycle $C_i$ such that $\mbox{wt}(C_i)$ has negative
slope and $\mbox{wt}(C_i)(x)$ is negative.

This proves the existence of $\alpha$ and $\beta$.  There are,
however, two bad cases that we wish to exclude.  Notice that if
$\alpha > \beta$ this means that for any evaluation at $x$, the
resulting graph has a negative weight cycle.  The same holds if there
is some cycle for which $\mbox{wt}(C_i)$ is constant and negative.  Let us
now show how $\alpha$ and $\beta$ can be efficiently computed whenever
these bad cases do not hold.  Indeed, $\alpha$ is the solution to the
following Linear Program (LP), which has a feasible solution if and
only if the bad cases do not hold.
\begin{center} \framebox{
\parbox[s]{3.0in}{\

{\bf Minimize $x$} under the constraints

\bigskip {\bf $\forall i$,  $\mbox{wt}(C_i)(x) \geq 0$}.

\bigskip }}
\end{center} This is an LP on one variable, but the number of
constraints can be exponential.
Using Megiddo's\cite{M-79}
technique for finding the minimum ratio cycles we can solve the
linear-program in $O(n^4\log n)$ steps.

\subsection{Proof of Lemma~\ref{lem:reweight}}

Let $\alpha$ and $\beta$ be the two numbers such that for all
$\alpha\leq r \leq \beta$ the graph $G(r)$ has no negative cycles and
at least one of $\alpha$ and $\beta$ is finite.

First let us consider the case when both $\alpha$ and $\beta$ are
finite. Recall that, given any number $r$, Johnson's algorithm
associates a weight function $h^{r}:V\rightarrow \mathbb{R}$ such
that, for any edge $(u,v)\in E$,
$$
W_{(u,v)}(r) + h^r(u) - h^r(v) \geq 0.
$$
(Johnson's algorithm computes this weight function by running the
Bellman-Ford algorithm over $G(r)$).  Define the weight function
$g^{[\alpha,\beta]}_v$ as
$$
g^{[\alpha, \beta]}_v(x) = \left(\frac{h^{\beta}(v) -
h^{\alpha}(v)}{\beta - \alpha}\right)x + h^{\alpha}(v) -
\left(\frac{h^{\beta}(v) - h^{\alpha}(v)}{\beta -
\alpha}\right)\alpha~.
$$
This is actually the equation of the line joining $(\alpha,
h^{\alpha}(v))$ and $(\beta, h^{\beta}(v))$ in $\mathbb{R}^2$.

\noindent Now we need to prove that for every $\alpha \leq r\leq
\beta$ and for every $(u,v)\in V$,
$$
W_{(u,v)}(r) + g^{[\alpha, \beta]}_{u}(r) - g^{[\alpha,\beta]}_{v}(r)
\geq 0~.
$$
Since $\alpha \leq r \leq \beta$, one can write $r = (1-\delta)\alpha
+ \delta\beta$ where $1 \geq \delta \geq 0$. Then for all $v\in V$,
$$
g^{[\alpha, \beta]}_v(r) = (1 - \delta) h^{\alpha}(v) + \delta
h^{\beta}(v)~.
$$
Since $W_{(u,v)}(r)$ is a linear function we can write
$$
W_{(u,v)}(r) = (1-\delta)W_{(u,v)}(\alpha) + \delta W_{(u,v)}(\beta)~.
$$
So after re-weighting the weight of the edge $(u,v)$ is
$$
(1-\delta)W_{(u,v)}(\alpha) + \delta W_{(u,v)}(\beta) +
(1-\delta)h^{\alpha}(u) + \delta h^{\beta}(u) - (1 - \delta)
h^{\alpha}(v) - \delta h^{\beta}(v)~.
$$
Now this is non-negative as by the definition of $h^{\beta}$ and
$h^{\alpha}$ we know that both $W_{(u,v)}(\beta) + h^{\beta}(u) -
h^{\beta}(v)$ and $W_{(u,v)}(\alpha) + h^{\alpha}(u) - h^{\alpha}(v)$
are non-negative.

We now consider the case when one of $\alpha$ or $\beta$ is not
finite.  We will prove it for the case where $\beta = +\infty$. The
case $\alpha = -\infty$ follows similarly.  Consider the simple
weighted graph $G_{\infty}=(V,E,W_{\infty})$ where the weight function
$W_{\infty}$ is defined as: if the weight of the edge $e$ is $W(e) =
a_ex + b_e$ then $W_{\infty}(e) = a_e$.

We run the Johnson's algorithm on the graph $G_{\infty}$. Let
$h^{\infty}(v)$ denote the weight that Johnson's algorithm associates
with the vertex $v$. Then define the weight function
$g^{[\alpha,\infty]}_v$ as
$$
g^{[\alpha,\infty]}_v(x) = h^{\alpha}(v) + (x-\alpha)h^{\infty}(v)~.
$$
We need to prove that for every $\alpha \leq r$ and for every
$(u,v)\in V$,
$$
W_{(u,v)}(r) + g^{[\alpha,\infty]}_{u}(r) - g^{[\alpha,
\infty]}_{v}(r) = W_{(u,v)}(r) + h^{\alpha}(u) + (r-\alpha)
h^{\infty}(u) - h^{\alpha}(v) - (r-\alpha) h^{\infty}(v) \geq 0~.
$$
Let $r = \alpha + \delta$ where $\delta \geq 0$. By the linearity of
$W$ we can write $W_{(u,v)}(r) = W_{(u,v)}(\alpha) + \delta
a_{(u,v)}$, where $W_{(u,v)}(r) = a_{(u,v)}r + b_{(u,v)}$. So the
above inequality can be restated as
$$
W_{(u,v)}(\alpha) + \delta a_{(u,v)} + h^{\alpha}(u) + \delta
h^{\infty}(u) - h^{\alpha}(v) - \delta h^{\infty}(v) \geq 0~.
$$
This now follows from the fact that both $W_{(u,v)}(\alpha) +
h^{\alpha}(u) - h^{\alpha}(v)$ and $a_{(u,v)} + h^{\infty}(u) -
h^{\infty}(v)$ are non-negative.

Since the running time of the reweighing part of Johnson's algorithm
takes $O(mn)$ time, the overall running time of computing the
functions $g^{[\alpha,\beta]}_v$ is $O(mn)$, as claimed.

\section{Proof of Theorem~\ref{thm:gen}}

In this section we construct a parametric algorithm that computes the
distance $\delta(u,v)$ between a given pair of vertices. If one is
interested in the actual path realizing this distance, then it can be
found with some extra book-keeping that we omit in the proof.

The processing algorithm will output the following advice: for any
pair $(u,v) \in V\times V$ the advice consists of a set of $t+2$ increasing
real numbers $-\infty = b_0 < b_1 < \dots < b_{t} < b_{t+1} = \infty$
and an ordered set of degree-$d$ polynomials $p_0, p_1, \dots, p_t$,
such that for all $b_i \leq r \leq b_{i+1}$ the weight of a shortest
path in $G(r)$ from $u$ to $v$ is $p_i(r)$.  Note that each $p_i$
corresponds to the weight of a path from $u$ to $v$. Thus if we are
interested in computing the exact path then we need to keep track of
the path corresponding to each $p_i$.

Given $r$, the instantiation algorithm has to find the $i$ such that
$b_i \leq r \leq b_{i+1}$ and then output $p_i(r)$. So the output
algorithm runs in time $O(\log t)$. To prove our result we need to
show that for any $(u,v) \in V\times V$ we can find the advice in time
$O(f(d) n)^{\log n}$. In particular this will prove that $t =
O(dn)^{\log n}$ and hence the result will follow.

\begin{definition} A \textit{minBase} is a sequence of increasing real
numbers $-\infty = b_0 < b_1 < \dots < b_t < b_{t+1} = \infty$ and an
ordered set of degree-$d$ polynomials $p_0, p_1, \dots, p_t$, such
that for all $b_i \leq r \leq b_{i+1}$ and all $j \neq i$, $p_i(r)
\leq p_j(r)$.
\end{definition} We call the sequence of real numbers the {\em
breaks}. We call each interval $[b_i, b_{i+1}]$ the $i$-th interval of
the minBase and the polynomial $p_i$ the $i$-th polynomial. The {\em
size} of the minBase is $t$.

The final advice that the preprocessing algorithm produces is a
minBase for every pair $(u,v)\in V\times V$ where the $i$-th
polynomial has the property that $p_i(r)$ is the distance from $u$ to
$v$ in $G(r)$ for each $b_i \le r \le b_{i+1}$.

\begin{definition} A $minBase^{\ell}(u,v)$ is a minBase corresponding
to the ordered pair $u,v$, where the $i$-th polynomial $p_i$ has the
property that for $r \in [b_i,b_{i+1}]$, $p_i(r)$ is the length of a
shortest path from $u$ to $v$ in $G(r)$, that is taken among all paths
that use at most $2^{\ell}$ edges.

\noindent A $minBase^{\ell}(u,w,v)$ is a minBase corresponding to the
ordered triple $(u,w, v)$ where the $i$-th polynomial $p_i$ has the
property that for each $r \in [b_i,b_{i+1}]$, $p_i(r)$ is the sum of
the lengths of a shortest path from $u$ to $w$ in $G(r)$, among all
paths that use at most $2^{\ell}$ edges, and a shortest path from $w$
to $v$ in $G(r)$, among all paths that use at most $2^{\ell}$ edges.
\end{definition} Note that in both of the above definitions some of
the polynomials can be $+\infty$ or $-\infty$.

\begin{definition} If $B_1$ and $B_2$ are two minBases (not
necessarily of the same size), with polynomials $p^1_i$ and $p^2_j$,
we say that another minBase with breaks $b'_k$ and polynomials $p'_k$
is $\min(B_1 + B_2)$ if the following holds.
\begin{enumerate}
\item For all $k$ there exist $i,j$ such that $p'_k = p^1_i +p^2_j$,
and
\item For $b'_k \leq r \leq b'_{k+1}$ and for all $i,j$ we have
$p'_k(r) \le p^1_i(r) + p^2_j(r)$.
\end{enumerate}
\end{definition}

\begin{definition} If $B_1,B_2,\ldots,B_s$ are $s$ minBases (not
necessarily of the same size), with polynomials $p^1_{i_1}, p^2_{i_2},
\ldots, p^s_{i_s}$, another minBase with breaks $b'_k$ and polynomials
$p'_k$ is $\min\{B_1,B_2,\ldots,B_s\}$ if the following holds.
\begin{enumerate}
\item For all $k$ there exist $q$ such that $p'_k = p^q_{i_q}$, and
\item For $b'_k \leq r \leq b'_{k+1}$ and for all $1 \le q \le s$ and
all $i_q$, we have $p'_k(r) \le p^q_{i_q}(r)$.
\end{enumerate}
\end{definition}

Note that using the above definition we can write the following two
equations:
\begin{equation}\label{eq:minuwv} minBase^{\ell+1}(u,v) = \min_{w\in
V}\left\{minBase^{\ell}(u,w,v)\right\}~.
\end{equation}

\begin{equation}\label{eq:minuwwv} minBase^{\ell}(u,w,v) =
min\left(minBase^{\ell}(u,w) + minBase^{\ell}(w,v)\right)~.
\end{equation}

The following claim will prove the result. The proof of the claim is
omitted due to lack of space.

\begin{claim}\label{cl:sizesminBase} If $B_1$ and $B_2$ are two
minBases of sizes $t_1$ and $t_2$ respectively, then
\begin{enumerate}
\item[(a)] $\min(B_1+B_2)$ can be computed from $B_1$ and $B_2$ in
time $O(t_1 + t_2)$.
\item[(b)] $\min\{B_1, B_2\}$ can be computed from $B_1$ and $B_2$ in
time $O(f(d)(t_1 + t_2))$, where $f(d)$ is the time required to
compute the intersection points of two degree-$d$ polynomials. The
size of $\min\{B_1, B_2\}$ is $O(d(t_1 + t_2))$.
\end{enumerate}
\end{claim} In order to compute $\min\{B_1,\ldots,B_s\}$ one
recursively computes $X=\min\{B_1,\ldots,B_{s/2}\}$ and
$Y=\min\{B_{s/2+1},\ldots,B_s\}$ and then takes $\min\{X,Y\}$.

If there are no negative cycles, then the advice that the
instantiation algorithm needs from the preprocessing algorithm
consists of $minBase^{\lceil\log n\rceil}(u,v)$. To deal with negative
cycles, both $minBase^{\lceil\log n\rceil}(u,v)$ and
$minBase^{\lceil\log n\rceil+1}(u,v)$ are produced, and the
instantiation algorithm compares them. if they are not equal, then the
correct output is $-\infty$.

Also note that $minBase^{0}(u,v)$ is the trivial minBase where the
breaks are $-\infty$ and $+\infty$ and the polynomial is weight
$W((u,v))$ associated to the edge $(u,v)$ if $(u,v) \in E$ and
$+\infty$ otherwise.

If the size of $minBase^{\ell}(u,v)$ is $s_{\ell}$, then by
(\ref{eq:minuwv}), (\ref{eq:minuwwv}), and by
Claim~\ref{cl:sizesminBase} the time to compute
$minBase^{\ell+1}(u,v)$ is $O(f(d))^{\log n}s_{\ell}$ and the size of
$minBase^{\ell+1}(u,v)$ is $O(d)^{\log n}s_{\ell}$. Thus one can
compute the advice for $u$ and $v$ in time
$$
(O(f(d))^{\log n})^{\log n} = O(n^{(O(1) + \log f(d))\log n})~,
$$
and the length of the advice string is $O(n^{(O(1) + \log d)\log n})$.

\section{Proof of Theorem \ref{thm:approx}}

Given the linear-weighted graph $G=(V,E,W)$, our preprocessing phase
begins by verifying that for all $r \in [\alpha,\beta]$, $G(r)$ has no
negative weight cycles. From the proof of Lemma \ref{lem:reweight} we
know that this holds if and only if both $G(\alpha)$ and $G(\beta)$
have no negative weight cycles. This, in turn, can be verified in
$O(mn)$ time using the Bellman-Ford algorithm.  We may now assume that
$G(r)$ has no negative cycles for any $r \in [\alpha,\beta]$.
Moreover, since our preprocessing algorithm will solve a large set of
shortest path problems, each of them on a specific instantiation of
$G$, we will first compute the reweighing functions
$g^{[\alpha,\beta]}_v$ of Lemma \ref{lem:reweight} which will enable
us to apply, in some cases, algorithms that assume nonnegative edge
weights.  Recall that by Lemma \ref{lem:reweight}, the functions
$g^{[\alpha,\beta]}_v$ for all $v \in V$ are computed in $O(mn)$ time.

The advice constructed by the preprocessing phase is composed of two
distinct parts, which we respectively call the {\em crude-short}
advice and the {\em refined-long} advice. We now describe each of
them.

For each edge $e \in E$, the weight is a linear function
$w_e=a_e+xb_e$.  Set $K=8(\beta-\alpha)\max_e|a_e|$. Let $N_0 = \lceil
K\sqrt{n}\ln n/\epsilon \rceil$ and let $N_1= \lceil Kn/\epsilon
\rceil$.  We define $N_0+1$ and $N_1+1$ points in
$[\alpha,\beta]$ and solve certain variants of shortest path problems
instantiated in these points.

Consider first the case of splitting $[\alpha,\beta]$ into $N_0$
intervals. Let $\rho_0 = (\beta-\alpha)/N_0$ and consider the points
$\alpha +i\rho_0$ for $i=0,\ldots,N_0$. The crude-short part of the
preprocessing algorithm solves $N_0+1$ {\em limited} all-pairs
shortest path problems in $G(\alpha +i\rho_0)$ for $i=0,\ldots,N_0$.
Set $t=4\sqrt{n}\ln n$, and let $d_i(u,v)$ denote the length of a
shortest path from $u$ to $v$ in $G(\alpha +i\rho_0)$ that is chosen
among all paths containing at most $t$ vertices (possibly
$d_i(u,v)=\infty$ if no such path exists).  Notice that $d_i(u,v)$ is
not necessarily the distance from $u$ to $v$ in $G(\alpha +i\rho_0)$,
since the latter may require more than $t$ vertices.  It is
straightforward to compute shortest paths limited to at most $k$
vertices (for any $1 \le k \le n$) in a real-weighted directed graph
with $n$ vertices in time $O(n^3 \log k)$ time, by the repeated
squaring technique. In fact, they can be computed in $O(n^3)$ time
(saving the $\log k$ factor) using the method from \cite{AHU-1974},
pp. 204--206.  This algorithm also constructs the predecessor data
structure that represents the actual paths.  It follows that for each
ordered pair of vertices $u,v$ and for each $i=0,\ldots,N_0$, we can
compute $d_i(u,v)$ and a path $p_i(u,v)$ yielding $d_i(u,v)$ in
$G(\alpha +i\rho_0)$ in $O(n^3|N_0|)$ time which is
$O(n^{3.5} \ln n)~.$
We also maintain, at no additional cost, linear functions $f_i(u,v)$
which sum the linear functions of the edges of $p_i(u,v)$. Note also
that if $d_i(u,v)=\infty$ then $p_i(u,v)$ and $f_i(u,v)$ are
undefined.

Consider next the case of splitting $[\alpha,\beta]$ into $N_1$
intervals. Let $\rho_1 = (\beta-\alpha)/N_1$ and consider the points
$\alpha +i\rho_1$ for $i=0,\ldots,N_1$. However, unlike the
crude-short part, the refined-long part of the preprocessing algorithm
cannot afford to solve an all-pairs shortest path algorithm for each
$G(\alpha +i\rho_1)$, as the overall running time will be too large.
Instead, we randomly select a set $H \subset V$ of (at most)
$\sqrt{n}$ vertices. $H$ is constructed by performing $\sqrt{n}$
independent trials, where in each trial, one vertex of $V$ is chosen
to $H$ uniformly at random (notice that since the same vertex can be
selected to $H$ more than once $|H| \le \sqrt{n}$).  For each $h \in
H$ and for each $i=0,\ldots,N_1$, we solve the single source shortest
path problem in $G(\alpha +i\rho_1)$ from $h$, and also (by reversing
the edges) solve the single-destination shortest path {\em toward}
$h$.  Notice that by using the reweighing functions
$g^{[\alpha,\beta]}_v$ we can solve all of these single source
problems using Dijkstra's algorithm.  So, for all $h \in H$ and
$i=0,\ldots,N_1$ the overall running time is
$$
O(|N_1||H|(m+n \log n)) = O(n^{1.5}m + n^{2.5}\log n) = O(n^{3.5})~.
$$
We therefore obtain, for each $h \in H$ and for each $i=0,\ldots,N_1$,
a shortest path tree $T_i(h)$, together with distances $d^*_i(h,v)$
from $h$ to each other vertex $v \in V$, which is the distance from
$h$ to $v$ in $G(\alpha +i\rho_1)$. We also maintain the functions
$f^*_i(h,v)$ that sum the linear equations on the path in $T^*_i(h)$
from $h$ to $v$.  Likewise, we obtain a ``reversed'' shortest path
tree $S^*_i(h)$, together with distances $d^*_i(v,h)$ from each $v \in
V$ to $h$, which is the distance from $v$ to $h$ in $G(\alpha
+i\rho_1)$. Similarly, we maintain the functions $f^*_i(v,h)$ that sum
the linear equations on the path in $S^*_i(h)$ from $v$ to $h$.

Finally, for each ordered pair of vertices $u,v$ and for each
$i=0,\ldots,N_1$ we compute a vertex $h_{u,v,i} \in H$ which attains
$
\min_{h \in H} d^*_i(u,h)+d^*_i(h,u)~.
$
Notice that the time to construct the $h_{u,v,i}$ for all ordered
pairs $u,v$ and for all $i=0,\ldots,N_1$ is $O(n^{3.5})$.  This
concludes the description of the preprocessing algorithm. Its overall
runtime is thus $O(n^{3.5} \ln n)$.

We now describe the instantiation phase. Given $u,v \in V$ and $r \in
[\alpha,\beta]$ we proceed as follows.  Let $i$ be the index for which
the number of the form $\alpha+i\rho_0$ is closest to $r$.  As we have
the advice $f_i(u,v)$, we let $w_0 = f_i(u,v)(r)$ (recall that
$f_i(u,v)$ is a function).  Likewise, let $j$ be the index for which
the number of the form $\alpha+j\rho_1$ is closest to $r$.  As we have
the advice $h=h_{u,v,j}$, we let $w_1 = f^*_j(u,h)(r)+f^*_j(h,u)(r)$.
Finally, our answer is $z=\min \{w_0,w_1\}$. Clearly, the
instantiation time is $O(1)$.  Notice that if we also wish to output a
path of weight $z$ in $G(r)$ we can easily do so by using either
$p_i(u,v)$, in the case where $z=w_0$ or using $S^*_j(h)$ and
$T^*_j(h)$ (we take the path from $u$ to $h$ in $S^*_j(h)$ and
concatenate it with the path from $h$ to $v$ in $T^*_j(h)$) in the
case where $z=w_1$.

It remains to show that, with very high probability, the result $z$
that we obtain from the instantiation phase is at most $\epsilon$
larger than the distance from $u$ to $v$ in $G(r)$.  For this purpose,
we first need to prove that the random set $H$ possesses some
``hitting set'' properties, with very high probability.


For every pair of vertices $u$ and $v$ and parameter $r$, let
$p_{u,v,r}$ be a shortest path in $G(r)$ among all simple paths from
$u$ to $v$ containing at least $t=4\sqrt{n}\ln n$ vertices (if $G$ is
strongly connected then such a path always exist, and otherwise we can
just put $+\infty$ for all $u,v$ pairs for which no such path exists).
The following simple lemma is used in an argument similar to one used
in \cite{Zw-2002}.
\begin{lemma}
\label{lem:prob} For fixed $u$, $v$ and $r$, with probability at least
$1-o(1/n^3)$ the path $p_{u,v,r}$ contains a vertex from $H$.
\end{lemma}
\begin{proof} Indeed, the path from $p_{u,v,r}$ by its definition has
at least $4\sqrt{n}\ln n$ vertices.  The probability that all of the
${\sqrt n}$ independent selections to $H$ failed to choose a vertex
from this path is therefore at most
$$
\left(1-\frac{4\sqrt{n} \ln n }{n} \right)^{\sqrt{n}} < e^{-4\ln n} <
\frac{1}{n^4} = o(1/n^3)~.
$$
\end{proof}

Let us return to the proof of Theorem \ref{thm:approx}.  Suppose that
the distance from $u$ to $v$ in $G(r)$ is $\delta$. We will prove that
with probability $1-o(1)$, $H$ is such that for every $u$, $v$ and $r$
we have $z \le \delta+\epsilon$ (clearly $z \ge \delta$ as it is the
precise length of some path in $G(r)$ from $u$ to $v$).
Assume first that there is a path $p$ of length
$\delta$ in $G(r)$ that uses less than $4\sqrt{n}\ln n$ edges.
Consider the length of $p$ in $G(\alpha+i\rho_0)$. When going from $r$
to $\alpha+i\rho_0$, each edge $e$ with weight $a_ex+b_e$ changed its
length by at most $|a_e|\rho_0$. By the definition of $K$, this is at
most $\rho_0 K/(8(\beta-\alpha))$.  Thus, $p$ changed its weight by at
most
$$
(4\sqrt{n} \ln n) \cdot \rho_0 \frac{K}{8(\beta-\alpha)} = (4\sqrt{n}
\ln n)\frac{K}{8N_0} < \frac{\epsilon}{2}.
$$
It follows that the length of $p$ in $G(\alpha+i\rho_0)$ is less than
$\delta+\epsilon/2$.  But $p_i(u,v)$ is a shortest path from $u$ to
$v$ in $G(\alpha+i\rho_0)$ of all the paths that contain at most $t$
vertices. In particular, $d_i(u,v) \le \delta+\epsilon/2$.  Consider
the length of $p_i(u,v)$ in $G(r)$. The same argument shows that the
length of $p_i(u,v)$ in $G(r)$ changed by at most $\epsilon/2$. But
$w_0=f_i(u,v)(r)$ is that weight, and hence $w_0 \le
\delta+\epsilon$. In particular, $z \le \delta+\epsilon$.

Assume next that every path of length $\delta$ in $G(r)$ uses at least
$4\sqrt{n}\ln n$ edges.  Let $p$ be one such path.  When going from
$r$ to $r'=\alpha+j\rho_1$, each edge $e$ with weight $a_ex+b_e$
changed its length by at most $|a_e|\rho_1$. By the definition of $K$,
this is at most $\rho_1 K/(8(\beta-\alpha))$.  Thus, $p$ changed its
weight by at most
$$
n \cdot \rho_1 \frac{K}{8(\beta-\alpha)} = n \frac{K}{8N_1} <
\frac{\epsilon}{8}.
$$
In particular, the length of $p_{u,v,r'}$ is not more than the length
of $p$ in $G(r')$, which, in turn, is at most $\delta+\epsilon/8$.  By
Lemma \ref{lem:prob}, with probability $1-o(1/n^3)$, some vertex of
$h$ appears on $p_{u,v,r'}$.  Moreover, by the union bound, with
probability $1-o(1)$ {\em all} paths of the type $p_{u,v,r'}$
(remember that $r'$ can hold one of $O(n)$ possible values) are thus
covered by the set $H$.  Let $h'$ be a vertex of $H$ appearing in
$p_{u,v,r'}$. We therefore have $d^*_j(u,h')+d^*_j(h',v) \le
\delta+\epsilon/8$. Since $h=h_{u,v,j}$ is taken as the vertex which
minimizes these sums, we have, in particular, $d^*_j(u,h)+d^*_j(h,v)
\le \delta+\epsilon/8$.  Consider the path $q$ in $G(\alpha+j\rho_1)$
realizing $d^*_j(u,h)+d^*_j(h,v)$.  The same argument shows that the
length of $q$ in $G(r)$ changed by at most $\epsilon/8$. But
$w_1=f^*_j(u,h)(r)+f^*_j(h,v)(r)$ is that weight, and hence $w_1 \le
\delta+\epsilon/4$. In particular, $z \le \delta+\epsilon/4$.

\section{Concluding remarks}\label{sec:conclusion} We have constructed
several parametric shortest path algorithms, whose common feature is
that they preprocess the generic instance and produce an advice that
enables particular instantiations to be solved faster than running the
standard weighted distance algorithm from scratch.  It would be of
interest to improve upon any of these algorithms, either in their
preprocessing time or in their instantiation time, or both.

Perhaps the most challenging open problem is to improve the
preprocessing time of Theorem \ref{thm:gen} to a polynomial one, or,
alternatively, prove an hardness result for this task.  Perhaps less
ambitious is the preprocessing time in Theorem \ref{thm:neg}.

Finally, parametric algorithms are of practical importance for other
combinatorial optimization problems as well. It would be interesting
to find applications where, indeed, a parametric algorithm can be
truly beneficial, as it is in the case of shortest path problems.

\section*{Acknowledgment} We thank Oren Weimann and Shay Mozes for
useful comments.

\end{document}